\makeatletter\@ifundefined{UseRawInputEncoding}{\newcommand\UseRawInputEncoding{}}{\UseRawInputEncoding}\makeatother
\documentclass{svproc}
\usepackage{url}

\usepackage{graphicx}
\usepackage{algorithm}
\usepackage{algpseudocode}
\usepackage{amsmath}
\usepackage{amssymb}
\usepackage{amsfonts}
\usepackage{hyperref}
\usepackage{enumitem}
\usepackage{bbding}
\usepackage{cite}
\hypersetup{pdfauthor={Name}}
\newtheorem{thm}{Theorem}
\newtheorem{prop}[thm]{Proposition}

\newcommand\doubleplus{\ensuremath{\mathbin{+\mkern-10mu+}}} 
\newcommand{\algrule}[1][.2pt]{\par\vskip.5\baselineskip\hrule height #1\par\vskip.5\baselineskip}
\newcommand*\ceq{\mathrel{\vcenter{\hbox{:}}{=}}}
\begin{document}
\mainmatter              
\title{{FRaGenLP}: A Generator of Random Linear Programming Problems for Cluster Computing Systems}

\titlerunning{Generator of LP Problems}
\author{Leonid B. Sokolinsky\Envelope \and Irina M. Sokolinskaya\thanks{The reported study was partially funded by the Russian Foundation for Basic Research (project No.~20-07-00092-a) and the Ministry of Science and Higher Education of the Russian Federation (government order FENU-2020-0022).}}
\authorrunning{L.B. Sokolinsky and I.M. Sokolinskaya} 
%
\tocauthor{Leonid B. Sokolinsky and Irina M. Sokolinskaya}
\institute{South Ural State University {(National Research University)} \\76, Lenin prospekt, Chelyabinsk, Russia, 454080\\
\email{leonid.sokolinsky@susu.ru}, \email{irina.sokolinskaya@susu.ru}}

\maketitle              

\begin{abstract}
The article presents and evaluates a scalable FRaGenLP algorithm for generating random linear programming problems of large dimension $n$ on cluster computing systems. To ensure the consistency of the problem and the boundedness of the feasible region, the constraint system includes $2n+1$ standard inequalities, called support inequalities. New random inequalities are generated and added to the system in a manner that ensures the consistency of the constraints. Furthermore, the algorithm uses two likeness metrics to prevent the addition of a new random inequality that is similar to one already present in the constraint system. The algorithm also rejects random inequalities that cannot affect the solution of the linear programming problem bounded by the support inequalities. The parallel implementation of the FRaGenLP algorithm is performed in C++ through the parallel BSF-skeleton, which encapsulates all aspects related to the MPI-based parallelization of the program. We provide the results of large-scale computational experiments on a cluster computing system to study the scalability of the FRaGenLP algorithm.

\end{abstract}

\section{Introduction}

The era of big data~\cite{sokol_1,sokol_2} has generated large-scale linear programming (LP) problems~\cite{sokol_3}. Such problems arise in economics, industry, logistics, statistics, quantum physics, and other fields. To solve them, high-performance computing systems and parallel algorithms are required. Thus, the development of new parallel algorithms for solving LP problems and the revision of current algorithms have become imperative. As examples, we can cite the works~\cite{sokol_4,sokol_5,sokol_6,sokol_7,sokol_8}. The development of new parallel algorithms for solving large-scale linear programming problems involves testing them on benchmark and random problems. One of the most well-known benchmark repositories of linear programming problems is Netlib-Lp~\cite{sokol_9}. However, when debugging LP solvers, it is often necessary to generate random LP problems with certain characteristics, with the dimension of the space and the number of constraints being the main ones.
The paper~\cite{sokol_10} suggested one of the first methods for generating random LP problems with known solutions. The method allows generating test problems of arbitrary size with a wide range of numerical characteristics. The main idea of the method is as follows. Take as a basis an LP problem with a known solution and then randomly modify it so that the solution does not change. The main drawback of the method is that fixing the optimal solution in advance significantly restricts the random nature of the resulting LP problem.

The article~\cite{sokol_11} describes the GENGUB generator, which constructs random LP problems with a known solution and given characteristics, such as the problem size, the density of the coefficient matrix, the degeneracy, the number of binding inequalities, and others.
A~distinctive feature of GENGUB is the ability to introduce generalized upper bound constraints, defined to be a (sub)set of constraints in which each variable appears at most once (i.e., has at most one nonzero coefficient). This method has the same drawback as the previous one: by preliminarily fixing the optimal solution, one significantly restricts the random nature of the resulting LP problem.

The article~\cite{sokol_12} suggests a method for generating random LP problems with a preselected solution type: bounded or unbounded, unique or multiple. Each of these structures is generated using random vectors with integer components whose range can be given. Next, an objective function that satisfies the required conditions, i.e., leads to a solution of the desired type, is obtained. The LP problem generator described in~\cite{sokol_12} is mainly used for educational purposes and is not suitable for testing new linear programming algorithms due to the limited variety of generated problems.

In the present paper, we suggest an alternative method for generating random LP problems. The method has the peculiarity of generating feasible problems of a given dimension with an unknown solution. The generated problem is fed to the input of the tested LP solver, and the latter outputs a solution that must be validated. The validator program (see, for example,~\cite{sokol_13}) validates the obtained solution. The method we suggest for generating random LP problems is named FRaGenLP (Feasible Random Generator of LP) and is implemented as a parallel program for cluster computing systems. The rest of the article is as follows. Section~\ref{sokol_Method} provides a formal description of the method for generating random LP problems and gives a sequential version of the FRaGenLP algorithm. In Section~\ref{sokol_Parallel_algorithm}, we discuss the parallel version of the FRaGenLP algorithm. In Section~\ref{sokol_Implementation}, we describe the implementation of FRaGenLP using a parallel BSF-skeleton and give the results of large-scale computational experiments on a cluster computing system. The results confirm the efficiency of our approach. Section~\ref{sokol_Conclusion} summarizes the obtained results and discusses plans to use the FRaGenLP generator in the development of an artificial neural network capable of solving large LP problems.

\section{Method for Generating Random LP Problems} \label{sokol_Method}

The method suggested in this paper generates random feasible bounded LP problems of arbitrary dimension~$n$ with an unknown solution. To guarantee the correctness of the LP problem, the constraint system includes the following \emph{support inequalities}:
\begin{equation}\label{sokol_Formula1}
\left\{ {\begin{array}{*{20}{l}}
  {{x_1}}&{}&{}&{}& \leqslant &\alpha  \\
  {}&{{x_2}}&{}&{}& \leqslant &\alpha  \\
  {}&{}& \ddots &{}& \cdots & \cdots  \\
  {}&{}&{}&{{x_n}}& \leqslant &\alpha  \\
  { - {x_1}}&{}&{}&{}& \leqslant &0 \\
  {}&{ - {x_2}}&{}&{}& \leqslant &0 \\
  {}&{}& \ddots &{}& \cdots & \cdots  \\
  {}&{}&{}&{ - {x_n}}& \leqslant &0 \\
  {{x_1}}&{ + {x_2}}& \cdots &{ + {x_n}}& \leqslant &{(n - 1)\alpha + \alpha/2}
  \end{array}} \right.
\end{equation}
Here, the positive constant $\alpha \in {\mathbb{R}_{ > 0}}$ is a parameter of the FRaGenLP generator. The number of support inequalities is $2n + 1$. The number of random inequalities is determined by a parameter~$d \in {\mathbb{Z}_{ \geqslant 0}}$. The total number $m$ of inequalities is defined by the following equation:
\begin{equation}\label{sokol_Formula2}m = 2n + 1 + d.\end{equation}
The coefficients of the objective function are specified by the vector
\begin{equation}\label{sokol_Formula3}c = \theta \left( {n,n - 1,n - 2, \ldots ,1} \right),\end{equation}
where the positive constant $\theta  \in {\mathbb{R}_{ > 0}}$ is a parameter of the FRaGenLP generator that satisfies the following condition:
\begin{equation}\label{sokol_Formula4}\theta  \leqslant \frac{\alpha }{2}.\end{equation}
From now on, we assume that the LP problem requires finding a feasible point at which the maximum of the objective function is attained. If the number $d$ of random inequalities is zero, then FRaGenLP generates an LP problem that includes only the support inequalities given in~\eqref{sokol_Formula1}. In this case, the LP problem has the following unique solution:
\begin{equation}\label{sokol_Formula5}\bar x = \left( {\alpha , \ldots ,\alpha ,{\alpha/2}} \right).\end{equation}

If the number $d$ of random inequalities is greater than zero, the FRaGenLP generator adds the corresponding number of inequalities to system~\eqref{sokol_Formula1}. The coefficients ${a_i} = \left( {{a_{i1}}, \ldots ,{a_{in}}} \right)$ of the random inequality and the constant term ${b_i}$ on the right side are calculated through the function $\operatorname{rand} (l,r)$, which generates a random real number in the interval $\left[ {l,r} \right]$ ($l,r \in \mathbb{R};l < r $), and the function $\operatorname{rsgn} ()$, which randomly selects a number from the set $\left\{ {1, - 1} \right\}$:
\begin{equation}\label{sokol_Formula6}\begin{array}{*{20}{c}}  {{a_{ij}} \ceq \operatorname{rsgn} () \cdot \operatorname{rand} (0,{a_{\mathrm{max}}}),} \\   {{b_i} \ceq \operatorname{rsgn} () \cdot \operatorname{rand} (0,{b_{\mathrm{max}}}).} \end{array}\end{equation}
Here, ${a_{\mathrm{max}}},{b_{\mathrm{max}}} \in {\mathbb{R}_{ > 0}}$ are parameters of the FRaGenLP generator. The inequality sign is always ``$\leqslant$''. Let us introduce the following notations:
\begin{equation}\label{sokol_Formula7}\operatorname{f} (x) = \langle {c,x} \rangle ;\end{equation}
\begin{equation}\label{sokol_Formula8}{h} = \left( {{\alpha  \mathord{\left/ {\vphantom {\alpha  2}} \right. \kern-\nulldelimiterspace} 2},{{ \ldots ,\alpha } \mathord{\left/ {\vphantom {{ \ldots ,\alpha } 2}} \right. \kern-\nulldelimiterspace} 2}} \right);\end{equation}
\begin{equation}\label{sokol_Formula9}\operatorname{dist}_h({a_i},{b_i}) = \frac{{\left| {\langle {{a_i},{h}} \rangle  - {b_i}} \right|}}{{\left\| {{a_i}} \right\|}};\end{equation}
\begin{equation}\label{sokol_Formula10}\pi (h,{a_i},{b_i}) = h - \frac{{\langle {{a_i},h} \rangle  - {b_i}}}{{|a_i|}^2}{a_i}.\end{equation}
Equation \eqref{sokol_Formula7} defines the objective function of the LP problem. Here and further on, $\langle { \cdot, \cdot } \rangle $ stands for the dot product of vectors. Equation \eqref{sokol_Formula8} defines the central point of the \emph{bounding hypercube} specified by the first $2n$ inequalities of system \eqref{sokol_Formula1}. Furthermore, Equation \eqref{sokol_Formula9} defines a function $\operatorname{dist}_h({a_i},{b_i})$ that gives the distance from the hyperplane $\langle {{a_i},x} \rangle  = {b_i}$ to the center~$h$ of the bounding hypercube. Here and below, $\left\|  {}\cdot{}  \right\|$ denotes the Euclidean norm. Equation \eqref{sokol_Formula10} defines a vector-valued function that expresses the orthogonal projection of the point $h$ onto the hyperplane $\langle {{a_i},x} \rangle  = {b_i}$.

To obtain a random inequality $\langle {{a_i},x} \rangle  \leqslant {b_i}$, we calculate the coordinates of the coefficient vector ${a_i}$ and the constant term ${b_i}$ using a pseudorandom rational number generator. The generated random inequality is added to the constraint system if and only if the following conditions hold:
\begin{equation}\label{sokol_Formula11}\langle {{a_i},{h}} \rangle  \leqslant {b_i};\end{equation}
\begin{equation}\label{sokol_Formula12}\rho  < \operatorname{dist}_h({a_i},{b_i}) \leqslant \theta ;\end{equation}
\begin{equation}\label{sokol_Formula13}\operatorname{f} \left( {\pi \left( {{h},{a_i},{b_i}} \right)} \right) > \operatorname{f} \left( {{h}} \right);\end{equation}
\begin{equation}\label{sokol_Formula14}\forall l \in \{ 1, \ldots ,i - 1\} :\neg \operatorname{like} ({a_i},{b_i},{a_l},{b_l}).\end{equation}
Condition~\eqref{sokol_Formula11} requires that the center of the bounding hypercube be a feasible point for the considered random inequality. If the condition does not hold, then the inequality $ - \langle {{a_i},x} \rangle  \leqslant  - {b_i}$ is added instead of $\langle {{a_i},x} \rangle  \leqslant {b_i}$. Condition~\eqref{sokol_Formula12} requires that the distance from the hyperplane $\langle {{a_i},x} \rangle  = {b_i}$ to the center~${h}$ of the bounding hypercube be greater than $\rho $ but not greater than $\theta $. The constant $\rho  \in {\mathbb{R}_{ > 0}}$ is a parameter of the FRaGenLP generator and must satisfy the condition $\rho  < \theta $, where $\theta $, in turn, satisfies condition \eqref{sokol_Formula4}. Condition~\eqref{sokol_Formula13} requires that the objective function value at the projection of the point ${h}$ onto the hyperplane $\langle {{a_i},x} \rangle  = {b_i}$ be greater than the objective function value at the point ${h}$. This condition combined with \eqref{sokol_Formula11} and \eqref{sokol_Formula12} cuts off constraints that cannot affect the solution of the LP problem. Finally, condition~\eqref{sokol_Formula14} requires that the new inequality be \emph{dissimilar} from all previously added ones, including the support ones. This condition uses the Boolean function ``$\operatorname{like}$'', which determines the \emph{likeness} of the inequalities $\langle {{a_i},x} \rangle  \leqslant {b_i}$ and $\langle {{a_l},x} \rangle  \leqslant {b_l}$ through the following equation:
\begin{equation}\label{sokol_Formula15}\operatorname{like} ({a_i},{b_i},{a_l},{b_l}) = \left\| {\frac{{{a_i}}}{{\left\| {{a_i}} \right\|}} - \frac{{{a_l}}}{{\left\| {{a_l}} \right\|}}} \right\| < {L_{\mathrm{max}}} \wedge \left| {\frac{{{b_i}}}{{\left\| {{a_i}} \right\|}} - \frac{{{b_l}}}{{\left\| {{a_l}} \right\|}}} \right| < {S_{\mathrm{min}}}.\end{equation}
The constants ${L_{\mathrm{max}}},{S_{\mathrm{min}}} \in {\mathbb{R}_{ > 0}}$ are parameters of the FRaGenLP generator. In this case, the parameter~${L_{\mathrm{max}}}$ must satisfy the condition
\begin{equation}\label{sokol_Formula16}{L_{\mathrm{max}}} \leqslant 0.7\end{equation}
(we will explain the meaning of this constraint below). According to~\eqref{sokol_Formula15}, inequalities $\langle {{a_i},x} \rangle  \leqslant {b_i}$ and $\langle {{a_l},x} \rangle  \leqslant {b_l}$ are \emph{similar} if the following two conditions hold:
\begin{equation}\label{sokol_Formula17}\left\| {\frac{{{a_i}}}{{\left\| {{a_i}} \right\|}} - \frac{{{a_l}}}{{\left\| {{a_l}} \right\|}}} \right\| < {L_{\mathrm{max}}};\end{equation}
\begin{equation}\label{sokol_Formula18}\left| {\frac{{{b_i}}}{{\left\| {{a_i}} \right\|}} - \frac{{{b_l}}}{{\left\| {{a_l}} \right\|}}} \right| < {S_{\mathrm{min}}}.\end{equation}

Condition~\eqref{sokol_Formula17} evaluates the measure of parallelism of the hyperplanes \linebreak
\mbox{$\langle {{a_i},x} \rangle  = {b_i}$} and $\langle {{a_l},x} \rangle  = {b_l}$,
which bound the feasible regions of the corresponding inequalities. Let us explain this. The unit vectors $e_i = a_i/\left\| a_i \right\|$ and $e_l = a_l/\left\| a_l \right\|$ are normal to the hyperplanes $\langle {{a_i},x} \rangle  = {b_i}$ and $\langle {{a_l},x} \rangle  = {b_l}$, respectively. Let us introduce the notation $\delta  = \left\| e_i - e_l\right\|$. If $\delta  = 0$, then the hyperplanes are parallel. If $0 \leqslant \delta  < {L_{\mathrm{max}}}$, then the hyperplanes are considered to be \emph{nearly parallel}.

Condition~\eqref{sokol_Formula18} evaluates the \emph{closeness} of the parallel hyperplanes $\langle {{a_i},x} \rangle  = {b_i}$ and $\langle {{a_l},x} \rangle  = {b_l}$. Indeed, the scalar values $\beta _i = b_i/\left\| a_i \right\|$ and  \linebreak $\beta _l = b_l/\left\| a_l \right\|$ are the normalized constant terms. Let us introduce the notation $\sigma  = | \beta _i - \beta _l |$. If $\sigma  = 0$, then the parallel hyperplanes coincide. If the hyperplanes are nearly parallel and $0 \leqslant \sigma  < S_{\mathrm{min}}$, then they are considered to be \emph{nearly concurrent}.

Two linear inequalities in $\mathbb{R}^n$ are considered \emph{similar} if the corresponding hyperplanes are nearly parallel and nearly concurrent.

The constraint \eqref{sokol_Formula16} for the parameter ${L_{\mathrm{max}}}$ is based on the following proposition.

\begin{prop}\label{sokol_Proposition1}
Let the two unit vectors $e,e' \in {\mathbb{R}^n}$ and the angle $\varphi  < \pi $ between them be given. Then,
\begin{equation}\label{sokol_Formula19}\left\| {e - e'} \right\| = \sqrt {2(1 - \cos \varphi )} .\end{equation}
\end{prop}

\begin{proof} By the definition of the norm in Euclidean space, we have
\[\begin{gathered}
  \left\| {e - e'} \right\| = \sqrt {\sum\limits_j {{{\left( {{e_j} - {{e'}_j}} \right)}^2}} }  = \sqrt {\sum\limits_j {\left( {{e_j}^2 - 2{e_j}{{e'}_j} + {{e'}_j}^2} \right)} }=  \hfill \\
   = \sqrt {\sum\limits_j {{e_j}^2}  - 2\sum\limits_j {{e_j}{{e'}_j}}  + \sum\limits_j {{{e'}_j}^2} }  = \sqrt {1 - 2\langle {{e_j},{{e'}_j}} \rangle  + 1} . \hfill \\
\end{gathered} \]
Thus,
\begin{equation}\label{sokol_Formula20}\left\| {e - e'} \right\| = \sqrt {2\left( {1 - \langle {{e_j},{{e'}_j}} \rangle } \right)} .\end{equation}
By the definition of the angle in Euclidean space, we have, for unit vectors,
\[\langle {{e_j},{{e'}_j}} \rangle  = \cos \varphi .\]
Substituting in \eqref{sokol_Formula20} the expression obtained, we have
\[\left\| {e - e'} \right\| = \sqrt {2\left( {1 - \cos \varphi } \right)} .\]
The proposition is proven.
\end{proof}

It is reasonable to consider that two unit vectors $e,e'$ are nearly parallel if the angle between them is less than $\pi/4$. In this case, according to \eqref{sokol_Formula19}, we have
\[\left\| {e - e'} \right\| < \sqrt {2\left( {1 - \cos \frac{\pi }{4}} \right)} .\]
Taking into account that $\cos(\pi/4) \approx 0.707$, we obtain the required estimate:
\[\left\| {e - e'} \right\| < 0.7.\]
\begin{figure}[t]
  \centering
  \includegraphics[scale=0.75]{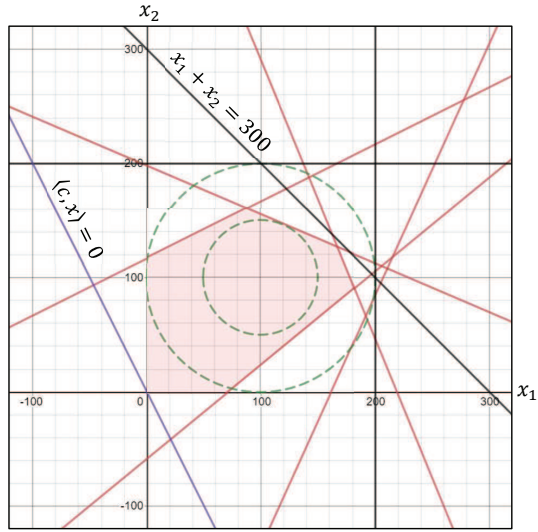}
  \caption{Random LP problem with $n=2$, $d=5$, $\alpha=200$, $\theta=100$, $\rho=50$,  $S_{\mathrm{min}}=100$, $L_{\mathrm{max}}=0.35$, $a_{\mathrm{max}}=1000$, and $b_{\mathrm{max}}=10\,000$}
  \label{sokol_Fig1}
\end{figure}

An example of a two-dimensional LP problem generated by FRaGenLP is shown in Figure~\ref{sokol_Fig1}. The purple color indicates the line defined by the coefficients of the objective function; the black lines correspond to the support inequalities, and the red lines correspond to the random inequalities. For the sake of clarity, we use green dashed lines to plot the large and the small circles defined by the equations ${\left( {{x_1} - 100} \right)^2} + {\left( {{x_2} - 100} \right)^2} = {100^2}$ and ${\left( {{x_1} - 100} \right)^2} + {\left( {{x_2} - 100} \right)^2} = {50^2}$. According to condition~\eqref{sokol_Formula12}, any random line must intersect the large green circle but not the small green circle. The semitransparent red color indicates the feasible region of the generated LP problem.
\begin{algorithm}[t]
\caption{Sequential algorithm for generating a random LP problem}\label{sokol_alg1}
\textbf{Parameters}: $n,d,\alpha,\theta, \rho, S_{\mathrm{min}},L_{\mathrm{max}},a_{\mathrm{max}},b_{\mathrm{max}}$
\algrule
\begin{algorithmic}[1]
\State $k \ceq 0$
\State $A \ceq [\:]$
\State $B \ceq [\:]$
\State $\operatorname{AddSupport}(A,B)$
\State \textbf{for} $j=n\ldots1$ \textbf{do} $c_j \ceq \theta\cdot j$
\State \textbf{if} $d=0$ \textbf{goto} 19
\State \textbf{for} $j=1\ldots n$ \textbf{do} $a_j \ceq \operatorname{rsign}( )\cdot \operatorname{rand}(0,a_{\mathrm{max}})$
\State $b \ceq \operatorname{rsign}( )\cdot \operatorname{rand}(0,b_{\mathrm{max}})$
\State \textbf{if} $\langle a,h\rangle\leqslant b$ \textbf{goto} 12
\State \textbf{for} $j=1\ldots n$ \textbf{do} $a_j \ceq -a_j$
\State $b \ceq -b$
\State \textbf{if} $\operatorname{dist}_h(a,b)<\rho$ \textbf{or} $\operatorname{dist}_h(a,b)>\theta$ \textbf{goto} 7
\State \textbf{if} $\operatorname{f}(\operatorname{\pi}(h,a,b))\leqslant\operatorname{f}(h)$ \textbf{goto} 7
\State \textbf{for} \textbf{all} $(\bar{a},\bar{b})\in(A,B)$ \textbf{do} \textbf{if} $\operatorname{like}(a,b,\bar{a},\bar{b})$ \textbf{goto} 7
\State $A \ceq A\doubleplus [a]$
\State $B \ceq B\doubleplus [b]$
\State $k \ceq k+1$
\State \textbf{if} $k<d$ \textbf{goto} 7
\State \textbf{output} $A,B,c$
\State \textbf{stop}
\end{algorithmic}
\end{algorithm}

Algorithm~1 represents a sequential implementation of the described method. Step~1 assigns zero value to the counter $k$ of random inequalities. Step~2 creates an empty list $A$ to store the coefficients of the inequalities. Step~3 creates an empty list $B$ to store the constant terms. Step~4 adds the coefficients and constant terms of the support inequalities \eqref{sokol_Formula1} to the lists $A$ and~$B$, respectively. Step~5 generates the coefficients of the objective function according to~\eqref{sokol_Formula3}. If the parameter $d$, which specifies the number of random inequalities, is equal to zero, then Step~6 passes the control to Step~19. Steps 7 and 8 generate the coefficients and the constant term of the new random inequality. Step~9 checks condition \eqref{sokol_Formula11}. If the condition does not hold, then the signs of the coefficients and the constant term are reversed (Steps 10,~11). Step~12 checks condition \eqref{sokol_Formula12}. Step~13 checks condition \eqref{sokol_Formula13}. Step~14 checks condition \eqref{sokol_Formula14}. Step~15 appends the coefficients of the new random inequality to the list~$A$ ($\doubleplus$ denotes the concatenation of lists). Step~16 appends the constant term of the new random inequality to the list~$B$. Step~17 increments the counter of added random inequalities by one. If the number of added random inequalities has not reached the given quantity $d$, then Step~18 passes the control to Step~7 to generate the next inequality. Step~19 outputs the results. Step~20 stops computations.

\section{Parallel Algorithm for Generating Random LP Problems}\label{sokol_Parallel_algorithm}

Implicit loops generated by passing the control from Steps 12--14 to Step~7 of Algorithm~1 can result in high overheads. For example, during the generation of the LP problem represented in Figure~1, there were 112\,581 returns from Step~12 to label~7, 32\,771 from Step~13, and 726 from Step~14. Therefore, generating a large random LP problem on a commodity personal computer can take many hours. To overcome this obstacle, we developed a parallel version of the FRaGenLP generator for cluster computing systems. This version is presented as Algorithm~2. It is based on the BSF parallel computation model~\cite{sokol_15,sokol_15_1}, which assumes the master--slave paradigm~\cite{sokol_16}. According to the BSF model, the master node serves as a control and communication center. All slave nodes execute the same code but on different data.
\begin{algorithm}[p]
\caption{Parallel algorithm for generating a random LP problem}\label{alg2}
\textbf{Parameters}: $n,d,\alpha,\theta, \rho, S_{\mathrm{min}},L_{\mathrm{max}},a_{\mathrm{max}},b_{\mathrm{max}}$
\algrule
\begin{multicols}{2}
\begin{center}
\textbf{Master} \\
\textbf{Slave (\emph{l}=1,\dots,\emph{L})}
\end{center}
\end{multicols}
\algrule
\begin{multicols}{2}
\begin{algorithmic}[1]
\State $k\ceq0$
\State $A_S\ceq[\:]$
\State $B_S\ceq[\:]$
\State $\operatorname{AddSupport}(A_S,B_S)$
\State \textbf{for} $j=n\ldots1$ \textbf{do} $c_j \ceq \theta\cdot j$
\State \textbf{output} $A_S,B_S,c$
\State \textbf{if} $d=0$ \textbf{goto} 36
\State $A_R\ceq[\:]$
\State $B_R\ceq[\:]$
\State
\State
\State
\State
\State
\State
\State
\State
\State \textbf{RecvFromSlaves} $a^{(1)},b^{(1)},...,a^{(L)},b^{(L)}$
\For{$l=1\ldots L$}
\State $isLike \ceq false$
\For{\textbf{all} $(\bar{a},\bar{b})\in(A_R,B_R)$}
\If{$\operatorname{like}(a^{(l)},b^{(l)},\bar{a},\bar{b})$}
\State$isLike \ceq true$
\State\textbf{goto} 27
\EndIf
\EndFor
\State\textbf{if} $isLike$ \textbf{continue}
\State$A_R \ceq A_R\doubleplus [a^{(l)}]$
\State$B_R \ceq B_R\doubleplus [b^{(l)}]$
\State$k \ceq k+1$
\State\textbf{if} $k=d$ \textbf{goto} 33
\EndFor
\State \textbf{SendToSlaves} $k$
\State \textbf{if} $k<d$ \textbf{goto} 18
\State \textbf{output} $A_R,B_R$
\State \textbf{stop}
\end{algorithmic}
\begin{algorithmic}[1]
\State \textbf{if} $d=0$ \textbf{goto} 36
\State $A_S\ceq[\:]$
\State $B_S\ceq[\:]$
\State $\operatorname{AddSupport}(A_S,B_S)$
\For{$j=1\ldots n$}
\State $a_j^{(l)} \ceq \operatorname{rsign}( )\cdot \operatorname{rand}(0,a_{\mathrm{max}})$
\EndFor
\State $b^{(l)} \ceq \operatorname{rsign}( )\cdot \operatorname{rand}(0,b_{\mathrm{max}})$
\State \textbf{if} $\langle a^{(l)},h\rangle\le b^{(l)}$ \textbf{goto} 12
\State \textbf{for} $j=1\ldots n$ \textbf{do} $a_j^{(l)} \ceq -a_j^{(l)}$
\State $b^{(l)} \ceq -b^{(l)}$
\State \textbf{if} $\operatorname{dist}_h(a^{(l)},b^{(l)})<\rho$ \textbf{goto} 5
\State \textbf{if} $\operatorname{dist}_h(a^{(l)},b^{(l)})>\theta$ \textbf{goto} 5
\State \textbf{if} $\operatorname{f}(\operatorname{\pi}(h,a^{(l)},b^{(l)}))\leqslant\operatorname{f}(h)$ \textbf{goto} 5
\For{\textbf{all} $(\bar{a},\bar{b})\in(A_S,B_S)$}
\State \textbf{if} $\operatorname{like}(a^{(l)},b^{(l)},\bar{a},\bar{b})$ \textbf{goto} 5
\EndFor
\State \textbf{SendToMaster} $a^{(l)},b^{(l)}$
\State
\State
\State
\State
\State
\State
\State
\State
\State
\State
\State
\State
\State
\State
\State \textbf{RecvFromMaster} $k$
\State \textbf{if} $k<d$ \textbf{goto} 5
\State
\State \textbf{stop}
\end{algorithmic}
\end{multicols}
\end{algorithm}

Let us discuss Algorithm~2 in more detail. First, we look at the steps performed by the master node. Step~1 assigns zero value to the counter of random inequalities. Step~2 creates an empty list ${A_S}$ to store the coefficients of the support inequalities. Step~3 creates an empty list ${B_S}$ to store the constant terms of the support inequalities. Step~4 adds the coefficients and constant terms of the support inequalities \eqref{sokol_Formula1} to the lists ${A_S}$ and~${B_S}$, respectively. Step~5 generates the coefficients of the objective function according to~\eqref{sokol_Formula3}. Step~6 outputs the coefficients and constant term of the support inequalities. If the parameter $d$, which specifies the number of random inequalities, is equal to zero, then Step~7 passes the control to Step~36, which terminates the computational process in the master node. Step~8 creates an empty list ${A_R}$ to store the coefficients of the random inequalities. Step~9 creates an empty list ${B_R}$ to store the constant terms of the random inequalities. In Step~18, the master node receives one random inequality from each slave node. Each of these inequalities satisfies conditions \eqref{sokol_Formula11}--\eqref{sokol_Formula13} and is not similar to any of the support inequalities. These conditions are ensured by the slave nodes. In the loop consisting of Steps~19--32, the master node checks all received random inequalities for similarity with the random inequalities previously included in the lists ${A_R}$ and ${B_R}$. The similar new inequalities are rejected, and the dissimilar ones are added to the lists ${A_R}$ and ${B_R}$. In this case, the inequality counter is increased by one each time some inequality is added to the lists. If the required number of random inequalities has already been reached, then Step~31 performs an early exit from the loop. Step~33 sends the current number of added random inequalities to the slave nodes. If this quantity is less than $d$, then Step~34 passes the control to Step~18, which requests a new portion of random inequalities from the slave nodes. Otherwise, Step~35 outputs the results, and Step~36 terminates the computational process in the master node.

Let us consider now the steps performed by the $l$-th slave node. If the parameter~$d$, which specifies the number of random inequalities, is equal to zero, then Step~1 passes the control to Step~36, which terminates the computational process in the slave node. Otherwise, Steps~2 and~3 create the empty lists ${A_S}$ and ${B_S}$ to store the support inequalities. Step~4 adds the coefficients and constant terms of the support inequalities \eqref{sokol_Formula1} to the lists ${A_S}$ and~${B_S}$, respectively. Steps 5--8 generate a new random inequality. Step~9 checks condition \eqref{sokol_Formula11}. If this condition does not hold, then the signs of the coefficients and the constant term are reversed (Steps~10 and~11). Steps~12--14 check conditions \eqref{sokol_Formula12} and \eqref{sokol_Formula13}. Steps~15--17 check the similarity of the generated inequality to the support inequalities. If any one of these conditions does not hold, then the control is passed to Step~5 to generate a new random inequality. If all conditions hold, then Step~18 sends the constructed random inequality to the master node. In Step~33  , the slave receives from the master the current number of obtained random inequalities. If this quantity is less than the required number, then Step~34 passes the control to Step~5 to generate a new random inequality. Otherwise, Step~36 terminates the computational process in the slave node.

\section{Software Implementation and the Computational Experiments}\label{sokol_Implementation}

We implemented the parallel Algorithm~2 in C++ through the parallel BSF-skele\-ton~\cite{sokol_17}, which is based on the BSF parallel computation model~\cite{sokol_15} and encapsulates all aspects related to the parallelization of the program using the MPI library~\cite{sokol_19}.

The BSF-skeleton requires the representation of the algorithm in the form of operations on lists using the higher-order functions \emph{Map} and \emph{Reduce}, defined by the Bird--Meertens formalism~\cite{sokol_20}. The required representation can be constructed as follows. Set the length of the \emph{Map} and \emph{Reduce} lists equal to the number of slave MPI processes. Define the \emph{Map} list items as empty structures: \[{\text{struct PT\_bsf\_mapElem\_T\{ \} }}{\text{.}}\]
Each element of the \emph{Reduce} list stores the coefficients and the constant term of one random inequality $\langle {a,x} \rangle \leqslant b$:
\[{\text{struct PT\_bsf\_reduceElem\_T\{ float a[n]; float b\} }}{\text{.}}\]

Each slave MPI process generates one random inequality using the \linebreak \emph{PC\_bsf\_MapF} function, which executes Steps 5--17 of Algorithm~2. The slave MPI process stores the inequality that satisfies all conditions to its local \emph{Reduce} list consisting of a single item. The master MPI process receives the generated elements from the slave MPI processes and places them in its \emph{Reduce} list (this code is implemented in the problem-independent part of the BSF-skeleton). After that, the master MPI process checks each obtained inequality for similarity with the previously added ones. If no matches are found, the master MPI process adds the inequality just checked to its local \emph{Reduce} list. These actions, corresponding to Steps~19--32 of Algorithm~2, are implemented as the standard function \emph{PC\_bsf\_ProcessResults} of the BSF-skeleton. The source code of the FRaGenLP parallel program is freely available on the Internet at \url{https://github.com/leonid-sokolinsky/BSF-LPP-Generator}.
\begin{table}[t]
\caption{Specifications of the ``Tornado SUSU'' computing cluster}
\centering
\begin{tabular}{l|l}
  \hline
  Parameter & Value \\
  \hline
  Number of processor nodes & 480 \\
  Processor & Intel Xeon X5680 (6 cores, 3.33 GHz) \\
  Processors per node & 2\\
  Memory per node & 24 GB DDR3\\
  Interconnect & InfiniBand QDR (40 Gbit/s) \\
  Operating system & Linux CentOS\\
  \hline
\end{tabular}\label{sokol_Table1}
\end{table}

Using the program, we conducted large-scale computational experiments on the cluster computing system ``Tornado SUSU''~\cite{sokol_21}. The specifications of the system are given in Table~\ref{sokol_Table1}. The computations were performed for several dimensions, namely $n = 3000$, $n = 5500$, and $n = 15\,000$. The total numbers of inequalities were, respectively, $6301$, $10\,001$, and $31\,501$. The corresponding numbers of random inequalities were $300$, $500$, and $1500$, respectively. Throughout the experiments, we used the following parameter values: $\alpha  = 200$ (the length of the bounding hypercube edge), $\theta  = 100$ (the radius of the large hypersphere), $\rho  = 50$ (the radius of the small hypersphere),  ${L_{\mathrm{max}}} = 0.35$ (the upper bound of \emph{near parallelism} for hyperplanes), ${S_{\mathrm{min}}} = 100$ (the minimum acceptable closeness for hyperplanes), ${a_{\mathrm{max}}} = 1000$ (the upper absolute bound for the coefficients),  and ${b_{\mathrm{max}}} = 10\,000$ (the upper absolute bound for the constant terms).

The results of the experiments are shown in Figure~\ref{sokol_Fig2}. Generating a random LP problem with $31\,501$ constraints with a configuration consisting of a master node and a slave node took 12 minutes. Generating the same problem with a configuration consisting of a master node and $170$ slave nodes took 22 seconds. The analysis of the results showed that the scalability bound (the maximum of the speedup curve) of the algorithm significantly depends on the dimension of the problem. For $n = 3000$, the scalability bound was 50 processor nodes approximately. This bound increased up to 110 nodes for $n = 5000$, and to 200 nodes for $n = 15\,000$. A further increase in problem size causes the processor nodes to run out of memory. It should be noted that the scalability bound of the algorithm significantly depends on the number of random inequalities too. Increasing this number by a factor of 10 resulted in a twofold reduction of the scalability bound. This is because an increase in the number of slave nodes results in a significant increase in the portion of sequential computations performed by the master node in Steps 19--32, during which the slave nodes are idle.
\begin{figure}[t]
  \centering
  \includegraphics[scale=1]{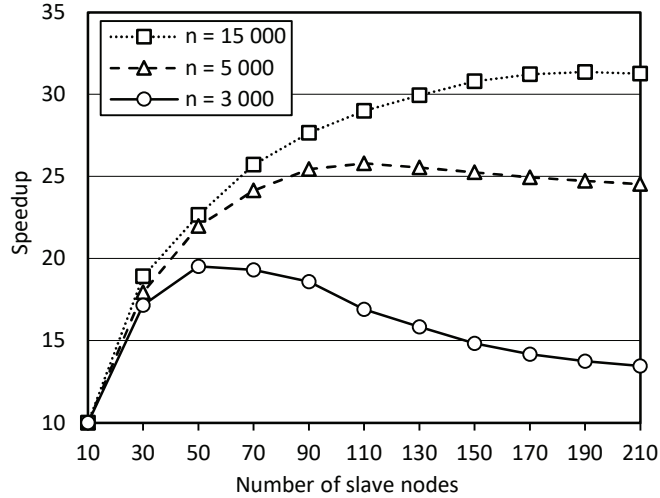}
  \caption{Speedup curves of the FRaGenLP parallel algorithm for various dimensions}
  \label{sokol_Fig2}
\end{figure}

\section{Conclusions}\label{sokol_Conclusion}

In this paper, we described the parallel FRaGenLP algorithm for generating random feasible bounded LP problems on cluster computing systems. In addition to random inequalities, the generated constraint systems include a standard set of inequalities called \emph{support inequalities}. They ensure the boundedness of the feasible region of the LP problem. In geometric terms, the feasible region of the support inequalities is a hypercube with edges adjacent to the coordinate axes, and the vertex that is farthest from the origin is cut off. The objective function is defined in such a manner that its coefficients decrease monotonically. The coefficients and constant terms of the random inequalities are obtained using a random number generator.
If the feasible region of a randomly generated inequality does not include the center of the bounding hypercube, then the sign of the inequality is reversed. Furthermore, not every random inequality is included in the constraint system. The random inequalities that cannot affect the solution of the LP problem for a given objective function are rejected. The inequalities, for which the bounding hyperplane intersects a small hypersphere located at the center of the bounding hypercube are also rejected. This ensures the feasibility of the constraint system. Moreover, any random inequality that is ``similar'' to at least one of the inequalities already added to the system (including the support ones) is also rejected. To define the ``similarity'' of inequalities, two formal metrics are introduced for bounding hyperplanes: the measure of parallelism and the measure of closeness.

The parallel algorithm is based on the BSF parallel computation model, which relies on the master--slave paradigm. According to this paradigm, the master node serves as a control and communication center. All slave nodes execute the same code but on different data. The parallel implementation was performed in C++ through the parallel BSF-skeleton, which encapsulates all aspects related to the MPI-based parallelization of the program. The source code of the \mbox{FRaGenLP} generator is freely available on the Internet at \linebreak \url{https://github.com/leonid-sokolinsky/BSF-LPP-Generator}.

Using this implementation, we conducted large-scale computational experiments on a cluster computing system. As the experiments showed, the parallel FRaGenLP algorithm demonstrates good scalability, up to 200 processor nodes for $n=15\,000$. Generating a random LP problem with $31\,501$ constraints takes 22 seconds with a configuration consisting of 171 processor nodes. Generating the same problem with a configuration consisting of a processor node takes 12 minutes. The program was used to generate a dataset of 70\,000 samples for training an artificial neural network capable of quickly solving large LP problems.

\end{document}